\documentclass[journal,onecolumn]{IEEEtranTCOM}
\normalsize

\usepackage{cite}

\ifCLASSINFOpdf
   \usepackage[pdftex]{graphicx}
\else
\fi

\usepackage[cmex10]{amsmath}
\usepackage{amssymb,mathalfa}
\usepackage{mathrsfs}
\usepackage{amsthm}

\usepackage{array}

\usepackage{mdwmath}
\usepackage{mdwtab}
\usepackage{url}

\font\msbm=msbm10 at 12pt
\newcommand{\ZZ}{\mbox{\msbm Z}}

\newcommand{\NN}{\mbox{\msbm N}}

\newcommand{\FF}{\mbox{\msbm F}}
\def \Z {{\ZZ}}

\def \N {{\NN}}
\def \F {{\FF}}

\newtheorem{theorem}{Theorem}
\newtheorem{lemma}[theorem]{Lemma}
\newtheorem{remark}[theorem]{Remark}

\newtheorem{definition}[theorem]{Definition}

\begin{document}
%
\title{On Code Rates of Fractional Repetition Codes}
%
%
%


\author{Krishna~Gopal~Benerjee,
        and~Manish~K~Gupta,~\IEEEmembership{Senior Member,~IEEE}
\thanks{Krishna Gopal Benerjee and Manish K Gupta are with Laboratory of Natural Information Processing, Dhirubhai Ambani Institute of Information and Communication Technology Gandhinagar, Gujarat, 382007, India, e-mail: (krishna\_gopal@daiict.ac.in and mankg@computer.org.).}
}

\maketitle

\begin{abstract}
In \textit{Distributed Storage Systems} (DSSs), usually, data is stored using replicated packets on different chunk servers. 
Recently a new paradigm of \textit{Fractional Repetition} (FR) codes have been introduced, in which, data is replicated in a smart way on distributed servers using a \textit{Maximum Distance Separable} (MDS) code. 
In this work, for a non-uniform FR code, bounds on the FR code rate and DSS code rate are studied. 
Using matrix representation of an FR code, some universally good FR codes have been obtained.
\end{abstract}


%
\IEEEpeerreviewmaketitle

\section{Introduction}\label{intro}
In the modern world, we generate tremendous amount of data daily and there is a pressing need to store
such a large amount of data in an efficient manner. Many companies such as Google, Microsoft, Dropbox, Amazon, etc. provide solution of data storage using Distributed Storage Systems (DSS). For a typical distributed storage environment \cite{5550492,DARKYC11}, a source file is stored on a large number of unreliable nodes such that each node has some encoded fraction of information of the source file called packets. By injecting redundancy, the system can repair the failed node, optimize the bandwidth and the storage need on each node etc.

Usually such a DSS is described by the parameters $(n,k,d)$, where a source file of size $B$ is divided into some encoded packets (from a finite field $\F_q$) of equal size. 
These encoded packets are distributed on $n$ distinct nodes in a smart way such that any $k$ (called reconstruction degree) nodes have sufficient encoded packets to reconstruct the complete source file $B$.
For reliability, the particular system has to have sufficient information to retrieve the complete file information in the presence of some failed nodes. 
The repair process is either exact (when the lost data is replaced with exact data) or functional (when the lost data is replaced with some function of lost data information). 
For a specific repair of a node failure, $\beta$ (repair traffic) packets are downloaded from each of $d$ (called repair degree) nodes and a new node is reconstructed with the downloaded packets which have the lost information. 
The particular new node will take place of the failed node and the repaired bandwidth $\gamma$ (the total number of packets downloaded to repair the node failure) will be $d\beta$. 
If the encoding of the message packets is done with repetition code then the system has minimum repair bandwidth though it needs more space to store data on a node. On the other hand, if one minimizes the storage space using Maximum Distance Separable (MDS) code on packet symbols then the system needs more repair bandwidth. 
In order to minimize the two conflicting parameters, Dimakis et al. proposed a new coding scheme called \textit{regenerating codes} \cite{5550492}. 
For an ($n,k,d$) DSS, the regenerating code is characterized by the parameters $[n,k,d,\alpha,\beta,B]$ over field $\F_q$, where $\alpha$ packets are stored on node $U_i$ ($\forall\; i\in[n]$\footnote{For $m\in\N,\ [m] = \{1,2,\ldots,m\}$.}). 
In \cite{5550492}, Dimakis et al. calculated min-cut bound for the regenerating codes in the presence of single node failure and plotted trade-off curve between node storage capacity $\alpha$ and repair bandwidth per node $\gamma$. 
\textit{Minimum storage regenerating} code (MSR code) and \textit{minimum bandwidth regenerating} code (MBR code) are obtained by minimizing the both parameters in a different order \cite{5550492}. 
The constructions for these two families are studied in \cite{5550492,DARKYC11,6620424,5961826,RSKR09,5513367} and \cite{5513263}. 
The decoding process requires some computations, so to avoid this, Rashmi et al. proposed a simple construction of MBR codes with uncoded repair on ($n,k,d=n-1$) DSS \cite{RSKR09}. 
In the particular MBR code, $B=kd-{k\choose 2}$, $\alpha$ = $n-1$. The construction is extended to new storage code called \textit{Fractional Repetition} (FR) code \cite{rr10}.

An FR code is a two layer coding, inner repetition code and outer MDS code \cite{rr10}. For an ($n,k,d$) DSS with identical repair traffic $\beta$, FR code can be characterized by the parameters $n,\theta,\alpha$ and $\rho$, where replicas of $\theta$ packets (each replicated $\rho$ times) are stored among $n$ distinct nodes having node storage capacity $\alpha$ each. Construction of an FR code in the work is done by Regular graph and Steiner system \cite{rr10}. The conditions on parameters for the existence of such FR codes are investigated in \cite{DBLP:journals/corr/abs-1201-3547}. Constructions of FR codes defined on ($n,k,d$) DSS are done by bipartite graphs \cite{6120326}, resolvable designs \cite{6483351,7422071}, Kronecker product of two Steiner systems \cite{6810361,7422071}, randomize algorithm \cite{6033980}, incidence matrix \cite{DBLP:journals/corr/abs-1303-6801} are also investigated. In \cite{6620277}, a construction using design theory is given for the FR code which achieves the bound (\cite[Theorem 4.1]{6846301}) on locality of regenerating codes. 

FR codes on heterogeneous DSS (DSS with some non-uniform parameters) are also investigated by researchers \cite{DBLP:journals/corr/abs-1302-3681,7118709,6804948,iet:/content/journals/10.1049/iet-com.2014.1225} and \cite{6763122}. In \cite{6804948}, \textit{Irregular Fractional Repetition code} (IFR code) is constructed by uniform hypergraph. In particular, IFR code can have distinct node storage capacity $\alpha_i$ ($\forall\; i\in[n]$) and a failed node can be repaired by some $d$ nodes. In \cite{6763122}, \textit{General Fractional Repetition} (GFR) code is constructed by group divisible designs. The specific GFR code is FR code with repair traffic $\beta\leq 1$ and dynamic node storage capacity $\alpha_i$ ($\forall\; i\in[n]$). In \cite{DBLP:journals/corr/abs-1302-3681}, FR code with dynamic repair degree called \textit{Weak Fractional Repetition code} (WFR code) is constructed by partial regular graph. FR code on Batch code is investigated in \cite{7118709}. In \cite{6811237}, \textit{Variable Fractional Repetition} (VFR) code is constructed in which replication factor of an arbitrary packet $P_j$ $(j\in[\theta])$ is either $\rho_1$ or $\rho_2$. In \cite{iet:/content/journals/10.1049/iet-com.2014.1225}, an FR code called \textit{Heterogeneous Fractional Repetitions} (HFR) code, with different replication factor of packets is discussed. Construction of FR code based on sequences is investigated in \cite{7458383} and \cite{DBLP:journals/corr/PrajapatiG16}. \textit{Flexible Fractional Repetitions} (FFR) code is constructed using pairwise balanced design and group divisible design  \cite{7066224}. \textit{Adaptive Fractional Repetitions} (Adaptive FR) code is constructed in \cite{7312417}, where Adaptive FR is an FR code generated by deleting elements from some known FR code. FR code is also constructed using perfect difference families and quasi-perfect difference families \cite{e18120441}. In \cite{7366761}, minimum distance on FR code is defined as the minimum number of nodes are such that those nodes fail simultaneously and the system can not recover the complete file. In \cite{7005805}, FR code is constructed by Turan graphs and combinatorial designs, called Transversal designs and generalized polygons. FR code with smaller repair degree called Locally Repairable Fractional Repetitions (Locally Repairable FR) is discussed in \cite{7458387,7558231}. Algorithms to compute repair degree and reconstruction degree for given packet distribution on nodes of an FR code, are given in \cite{DBLP:journals/corr/abs-1305-4580}. 

\textit{Contribution:} 
In the present work, FR code with non-uniform parameters such as node storage capacity $\alpha_i$ ($i\in[n]$), repair traffic $\beta_i$, repair degree $d_i$ and replication factor $\rho_j$ ($j\in[\theta]$), is considered. 
The \textit{FR code rate} for an FR code is the ratio of file size to the total number of encoded packets stored in the FR code. 
Bounds on the FR code rate, growth of the FR code rate and the DSS code rate are obtained in this work. It is shown that the FR code with FR code rate more than the DSS code rate, does not exist. 
Again it is proved that the reliable FR code with FR code rate more than $0.5$, does not exist. 
A reliable FR code with the maximum code rate $0.50$ has the lowest tolerance factor $1$.

\textit{Organization:} In Section $2$, we introduced FR code and obtained some of its basic properties. 
Some bounds on the DSS code rate and the FR code rate for an FR code is obtained in Section $3$. 
Towards the end of the section, some constructions of the universally good FR codes are also given. 
Section $4$ concludes the paper with some general remarks.

\section{Background on Fractional Repetition Codes}

An FR code is a distribution of replicated packets on nodes such that the system is reliable and the file can be retrieved. Preliminaries of an FR code is collected in the section. 

Consider a file divided into $B$ distinct independent message packets. The symbols associated with the message packets, are defined on a field $\F_q$. 
By utilizing ($\theta$, $B$) MDS code on the message packets, $\theta$ distinct encoded packets $P_1$, $P_2$, $\ldots$, $P_{\theta}$ are generated. 
A packet $P_j$ ($j\in[\theta]$) is replicated $\rho_j$ times and all the replicated packets are distributed among $n$ distinct nodes $U_1$, $U_2$, $\ldots$, $U_n$. Let $\rho$ = $\max\{\rho_j:j\in[\theta]\}$ be the maximum replication factor among all packets. 
If $\alpha_i$ $(\forall \ i\in[n])$ packets are stored on the node $U_i$ then the parameter $\alpha_i$ can be called storage capacity of node $U_i$. 
The maximum and minimum node storage capacity among all nodes can be given by $\alpha$ = $\max\{\alpha_i:i\in[n]\}$ and $\alpha_{min} = \min\{\alpha_i:i\in[n]\}$ respectively. 
An example of such FR code $\mathscr{C}:(n=6,\theta=10,\alpha=4,\rho=3)$ is shown in Table \ref{example}. 
In the table, a file is divided into $7\ (=B)$ message packets and the $7$ message packets are encoded into $10\ (=\theta)$ distinct packets (say $P_1$, $P_2$, $\ldots$, $P_{10}$) by using $(10, 7)$ MDS code. 
For the FR code, the packet $P_8$ is replicated thrice and the remaining packets are replicated twice $i.e.$ $\rho_8=3$ and $\rho_j=2$ ($\forall j\in[10]\backslash\{8\}$). 
All $21$ replicated packets are distributed on $6\ (=n)$ distinct nodes say $U_1$, $U_2$, $\ldots$, $U_6$ such that $U_1$ = $\{P_1,P_2,P_3,P_4\}$, $U_2$ = $\{P_1,P_5,P_6,P_7\}$, $U_3$ = $\{P_2,P_5,P_8,P_9\}$, $U_4$ = $\{P_8,P_6,P_8\}$, $U_5$ = $\{P_4,P_7,P_{10}\}$ and $U_6$ = $\{P_8,P_9,P_{10}\}$. 
In the example, $\alpha_i$ = $|U_i|$ = $4$ ($i=1,2,3$) and $\alpha_i$ = $|U_i|$ = $3$ ($i=4,5,6$). 
Observe that $\rho=\max\{\rho_j:j\in[10]\}=3$, $\alpha=\max\{\alpha_i:i\in[6]\}=4$ and $\alpha_{min}=\min\{\alpha_i:i\in[6]\}=3$.

\begin{table}[ht]
\caption{Distribution of replicated packets on nodes for the FR code $\mathscr{C}:(6,10,4,3)$.}
\centering 
\begin{tabular}{|c||c|}
\hline
$U_1$& $P_1$, $P_2$, $P_3$, $P_4$         \\
\hline \hline
$U_2$& $P_1$, $P_5$, $P_6$, $P_7$         \\
\hline \hline
$U_3$& $P_2$, $P_5$, $P_8$, $P_9$         \\
\hline \hline 
$U_4$& $P_3$, $P_6$, $P_8$                \\
\hline \hline
$U_5$& $P_4$, $P_7$, $P_{10}$             \\
\hline \hline
$U_6$& $P_8$, $P_9$, $P_{10}$             \\
\hline  
\end{tabular}
\label{example}
\end{table}
Using MDS property, a data collector can reconstruct the file by collecting any $B$ distinct packets out of the $\theta$ packets from the system. To download $B$ distinct packets, the data collector has to connect some nodes out of the $n$ nodes. A set of those nodes can be called as \textit{reconstruction set} and the set can be denoted by $\mathcal{A}$. Clearly, $|\mathcal{A}|$ is reconstruction degree and $B\leq |\cup_{U\in\mathcal{A}} U|$. If an arbitrary node has some packets which are not stored in any other node in the reconstruction set $\mathcal{A}$ 
then the data collector has to download packets from each node of the reconstruction set $\mathcal{A}$.
In other words, for any proper subset $A\subsetneq\mathcal{A}$, if $B>|\cup_{U\in A}U|$ then each node in $\mathcal{A}$ will take active participation during file reconstruction by the data collector. 
In particular, removing one node from the node collection, the resultant collection does not have $B$ distinct packets. 
Note that reconstruction degrees may not be identical for distinct reconstruction sets. 
All possible reconstruction sets for FR code $\mathscr{C}:(6,10,4,3)$ (see Table \ref{example}) are listed in Table \ref{list}. 
For an FR code, if $\eta\in\N$ distinct reconstruction sets exist, then the maximum reconstruction degree will be $k_{max}:=\max\{|\mathcal{A}_i|:i\in[\eta]\}$. 
In the FR code, data collector can get the complete file by downloading packets from any $k_{max}$ nodes. 
In the Table \ref{list}, $k_{max}=3$ so any three nodes in FR code $\mathscr{C}:(6,10,4,3)$ have sufficient packets for the complete file $B$.

\begin{table}[ht]
\caption{Reconstruction sets and respective reconstruction degrees for the FR code $\mathscr{C}:(6,10,4,3)$.}
\centering 
\begin{tabular}{|c|c|}
\hline
Reconstruction Set        & Reconstruction Degree      \\
\hline \hline
$\mathcal{A}_1=\{U_1,U_2\}$   & $2$         \\
\hline 
$\mathcal{A}_2=\{U_1,U_3\}$   & $2$         \\
\hline 
$\mathcal{A}_3=\{U_1,U_4,U_5\}$ & $3$         \\
\hline 
$\mathcal{A}_4=\{U_2,U_3\}$   & $2$         \\
\hline 
$\mathcal{A}_5=\{U_2,U_4,U_5\}$ & $3$         \\
\hline 
$\mathcal{A}_6=\{U_2,U_6\}$   & $2$         \\
\hline 
$\mathcal{A}_7=\{U_3,U_5\}$   & $2$         \\
\hline 
$\mathcal{A}_8=\{U_3,U_4,U_6\}$ & $3$         \\
\hline 
$\mathcal{A}_9=\{U_1,U_6\}$   & $2$         \\
\hline 
\end{tabular}
\label{list}
\end{table}

For an FR code $\mathscr{C}:(n,\theta,\alpha,\rho)$, the system can be repaired by replacing failed node $U_i\ (i\in[n])$ with newly generated node $U_i$'. Packets stored on the new node $U_i$', are downloaded from some remaining non-failure nodes called \textit{helper nodes}. For a node failure $U_i$, set of such helper nodes is called surviving set and denoted by $S_i^{(\ell)}$, where $\ell\leq\prod\limits_{j;P_j\in U_i}(\rho_j-1)$ is the index of surviving sets. Clearly, $U_i\subseteq\cup_{U\in S_i^{(\ell)}}U$ $i.e.$ $\alpha_i\leq |\cup_{U\in S_i^{(\ell)}}U|$. For $U_r\in S_i^{(\ell)}$, if $U_i\cap U_r = \phi$ then $U_r$ does not have copy of packets which are stored in node $U_i$. Hence, for any $U_r\in S_i^{(\ell)}$, $U_i\cap U_r \neq \phi$. Packets which are seared by node $U_r\in S_i^{(\ell)}$ for node $U_i$', can not be seared by any other node in $S_i^{(\ell)}\backslash\{U_r\}$ so $U_i\nsubseteq \bigcup\limits_{U\in S_i^{(\ell)}\backslash\{U_r\}} U$. Mathematically, $S_i^{(\ell)}=\{U_r: r\in I\subset [n]\backslash\{i\}, U_i\cap U_r\neq\phi\}$ such that node $U_i\subseteq \bigcup\limits_{m\in I}U_m$ and $U_i\nsubseteq \bigcup\limits_{m\in S\subset I}U_m$. For example, if the node $U_1$ fails in the FR code $\mathscr{C}:(6,10,4,3)$ (considered in Table \ref{example}) then a new node $U_1'$ is created by downloading packets from nodes $U_2$, $U_3$, $U_4$ and $U_5$. Hence, the set $S_1^{(1)}=\{U_2,U_3,U_4,U_5\}$ is surviving set for node $U_1$. Surviving sets for nodes in the FR code $\mathscr{C}:(6,10,4,3)$ (considered in Table \ref{example}), are given in Table \ref{surviving set}. There may be several surviving sets for a specific node with different cardinality (for example, see Table \ref{surviving set}). Observe, the repair degree of a failed node $U_i$, is the cardinality of surviving set $S_i^{(\ell)}$ which is used for repair process at that time instant. In the FR code $\mathscr{C}:(6,10,4,3)$, the repair degrees of nodes are listed in Table \ref{surviving set}. For a given node $U_i$, the repair degree can be distinct at different time instant. The maximum repair degree for a node failure $U_i$ can be given by $d_i=\max\{|S_i^{(\ell)}|:\ell\leq\prod\limits_{j;P_j\in U_i}(\rho_j-1),\ell\in\Z\}$. For an FR code $\mathscr{C}:(n,\theta,\alpha,\rho)$, let $d$ = $\max\{d_i:i\in[n]\}$ be the maximum repair degree for the FR code.

\begin{table}[ht]
\caption{Surviving sets and repair degrees for nodes in the FR Code $\mathscr{C}:(6,10,4,3)$.}
\centering 
\begin{tabular}{|c||c|c|c|c|}
\hline
Node &Surviving Set                   & Index set     & Repair Degree     & Max. Repair      \\
$U_i$&$S_i^{(\ell)}$                  & $I$           &$|S_i^{(\ell)}|$   & Degree $d_i$ \\
\hline \hline
$U_1$& $S_1^{(1)}=$                   & $\{2,3,4,5\}$ & $4$               & $4$  \\ 
     & $\{U_2,U_3,U_4,U_5\}$          &               &                   &      \\ 
\hline 
$U_2$&$S_2^{(1)}=$                    & $\{1,3,4,5\}$ & $4$               & $4$  \\
     &$\{U_1,U_3,U_4,U_5\}$           &               &                   &      \\
\hline 
$U_3$&$S_3^{(1)}=$                    & $\{1,2,6\}$   & $3$               & $4$  \\ 
     &$\{U_1,U_2,U_6\}$               &               &                   &      \\ \cline{2-4}
     &$S_3^{(2)}=$                    & $\{1,2,4,6\}$ & $4$               &      \\
     &$\{U_1,U_2,U_4,U_6\}$           &               &                   &      \\
\hline 
$U_4$&$S_4^{(1)}=$                    & $\{1,2,3\}$   & $3$               & $3$  \\ 
     &$\{U_1,U_2,U_3\}$               &               &                   &      \\ \cline{2-4}
     &$S_4^{(2)}=$                    & $\{1,2,6\}$   & $3$               &      \\
     &$\{U_1,U_2,U_6\}$               &               &                   &      \\
\hline 
$U_5$&$S_5^{(1)}=$                    & $\{1,2,6\}$   & $3$               & $3$  \\
     &$\{U_1,U_2,U_6\}$               &               &                   &      \\
\hline 
$U_6$&$S_6^{(1)}=\{U_3,U_5\}$         & $\{3,5\}$     & $2$               & $3$  \\ \cline{2-4}
     &$S_6^{(2)}=$                    & $\{3,4,5\}$   & $3$               &      \\
     &$\{U_3,U_4,U_5\}$               &               &                   &      \\
\hline 
\end{tabular}
\label{surviving set}
\end{table}

At a time instant $t$, assume a failed node $U_i$ is repaired by replacing it with newly generated node $U_i$' by using nodes of surviving set $S_i^{(\ell)}$. In particular process, node $U_r\in S_i^{(\ell)}$ downloads $\beta\left(U_i,U_r,S_i^{(\ell)}\right)>0$ number of packets and stores on node $U_i$'. For a node failure in the FR code $\mathscr{C}:(6,10,4,3)$ (as considered in Table \ref{example}), the downloaded packets are listed in the Table \ref{beta}.

\begin{table}
\caption{Repair of each failed node $U_i$ in the FR code $\mathscr{C}:(6,10,4,3)$.}
\centering 
\begin{tabular}{|c||c|c|c|c|}
\hline
Node&Surviving                     & Helper    & Downloaded      & Number of Packets    \\
$U_i$& Set $S_i^{(\ell)}$          & Node $U_r$& Packets         & Downloaded by $U_r$ \\
\hline \hline 
     &                &  $U_2$        &$P_1$      & $\beta\left(U_1,U_2,S_1^{(1)}\right)=1$  \\ \cline{3-5}
$U_1$&   $S_1^{(1)}$  &  $U_3$        &$P_2$      & $\beta\left(U_1,U_3,S_1^{(1)}\right)=1$  \\ \cline{3-5}
     &                &  $U_4$        &$P_3$      & $\beta\left(U_1,U_4,S_1^{(1)}\right)=1$  \\ \cline{3-5}
     &                &  $U_5$        &$P_4$      & $\beta\left(U_1,U_5,S_1^{(1)}\right)=1$  \\ 
\hline \hline  
     &                &  $U_1$        &$P_1$      & $\beta\left(U_2,U_1,S_2^{(1)}\right)=1$  \\ \cline{3-5}
$U_2$&   $S_2^{(1)}$  &  $U_3$        &$P_5$      & $\beta\left(U_2,U_3,S_2^{(1)}\right)=1$  \\ \cline{3-5}
     &                &  $U_4$        &$P_6$      & $\beta\left(U_2,U_4,S_2^{(1)}\right)=1$  \\ \cline{3-5}
     &                &  $U_5$        &$P_7$      & $\beta\left(U_2,U_5,S_2^{(1)}\right)=1$  \\ 
\hline  \hline 
     &                &  $U_1$        &$P_2$      & $\beta\left(U_3,U_1,S_3^{(1)}\right)=1$  \\ \cline{3-5}
$U_3$&$S_3^{(1)}$     &  $U_2$        &$P_5$      & $\beta\left(U_3,U_2,S_3^{(1)}\right)=1$  \\ \cline{3-5}
     &                &  $U_6$        &$P_8$, $P_9$& $\beta\left(U_3,U_6,S_3^{(1)}\right)=2$  \\ \cline{2-5}
     &                &  $U_1$        &$P_2$      & $\beta\left(U_3,U_1,S_3^{(2)}\right)=1$  \\ \cline{3-5}
       &$S_3^{(2)}$     &  $U_2$        &$P_5$      & $\beta\left(U_3,U_2,S_3^{(2)}\right)=1$  \\ \cline{3-5}
     &                &  $U_4$        &$P_8$      & $\beta\left(U_3,U_4,S_3^{(2)}\right)=1$  \\ \cline{3-5}
     &                &  $U_6$        &$P_9$      & $\beta\left(U_3,U_6,S_3^{(2)}\right)=1$  \\ 
\hline  \hline 
     &                & $U_1$         &$P_3$      & $\beta\left(U_4,U_1,S_4^{(1)}\right)=1$  \\ \cline{3-5}
$U_4$&$S_4^{(1)}$     & $U_2$         &$P_6$      & $\beta\left(U_4,U_2,S_4^{(1)}\right)=1$  \\ \cline{3-5}
     &                & $U_3$         &$P_8$      & $\beta\left(U_4,U_3,S_4^{(1)}\right)=1$  \\ \cline{2-5}
     &                & $U_1$         &$P_3$      & $\beta\left(U_4,U_1,S_4^{(2)}\right)=1$  \\ \cline{3-5}
         &$S_4^{(2)}$     & $U_2$         &$P_6$      & $\beta\left(U_4,U_2,S_4^{(2)}\right)=1$  \\ \cline{3-5}
     &                & $U_6$         &$P_8$      & $\beta\left(U_4,U_6,S_4^{(2)}\right)=1$  \\ 

\hline  \hline 
     &                & $U_1$         &$P_4$      & $\beta\left(U_5,U_1,S_5^{(1)}\right)=1$  \\ \cline{3-5}
$U_5$&$S_5^{(1)}$     & $U_2$         &$P_7$      & $\beta\left(U_5,U_2,S_5^{(1)}\right)=1$  \\ \cline{3-5}
     &                & $U_6$         &$P_{10}$   & $\beta\left(U_5,U_6,S_5^{(1)}\right)=1$  \\ 
\hline  \hline 
     &$S_6^{(1)}$     & $U_3$         &$P_8$, $P_9$& $\beta\left(U_6,U_3,S_6^{(1)}\right)=2$  \\ \cline{3-5}
     &                & $U_5$         &$P_{10}$   & $\beta\left(U_6,U_5,S_6^{(1)}\right)=1$  \\ \cline{2-5}
$U_6$&                & $U_3$         &$P_9$      & $\beta\left(U_6,U_3,S_6^{(2)}\right)=1$  \\ \cline{3-5}
     &$S_6^{(2)}$     & $U_4$         &$P_8$      & $\beta\left(U_6,U_4,S_6^{(2)}\right)=1$  \\ \cline{3-5}
     &                & $U_5$         &$P_{10}$   & $\beta\left(U_6,U_5,S_6^{(2)}\right)=1$  \\ 
\hline 
\end{tabular}
\label{beta}
\end{table}

Now the \textit{repair bandwidth} of a failed node $U_i$, is the total number of packets downloaded by each helper nodes of a surviving set $S_i^{(\ell)}$ during repair process. Mathematically, the repair bandwidth for a failed node $U_i$ is given by 
\begin{equation*}
\gamma\left(U_i, S_i^{(\ell)}\right)=\sum_{\stackrel{r}{U_r\in S_i^{(\ell)}}}\beta\left(U_i, U_r, S_i^{(\ell)}\right).
\end{equation*}
For a node failure, the repair bandwidths which are associated with different surviving sets may not be same.
For the FR code $\mathscr{C}:(6,10,4,3)$, one can observe $\gamma\left(U_1, S_1^{(1)}\right)$ = $\gamma\left(U_2, S_2^{(1)}\right)$ = $\gamma\left(U_3, S_3^{(1)}\right)$ = $\gamma\left(U_3, S_3^{(2)}\right)$ = $4$ and $\gamma\left(U_4, S_4^{(1)}\right)$ = $\gamma\left(U_4, S_4^{(2)}\right)$ = $\gamma\left(U_5, S_5^{(1)}\right)$ = $\gamma\left(U_6, S_6^{(1)}\right)$ = $\gamma\left(U_6, S_6^{(2)}\right)$ = $3$ by using the Table \ref{beta}. Note that an FR code $\mathscr{C}:(n,\theta,\alpha,\rho)$ can tolerate $\min\{\rho_j:j\in[\theta]\}-1$ node failure simultaneously.

Now, one can define an FR code as follows.

\begin{definition} (Fractional Repetition Code): 
 On a Distributed Storage System with $n$ nodes denoted by $U_i\ (i \in[n])$ and $\theta$ packets denoted by $P_j\ ( j \in [\theta])$, one can define an FR code $\mathscr{C}(n, \theta, \alpha, \rho)$ as a collection $\mathscr{C}$ of $n$ subsets $U_i$ ($i\in[n]$) of a set $\{P_j:j\in[\theta]\}$, such that an arbitrary packet $P_j$ is the element of exactly $\rho_j$ ($\in\N$) distinct subsets in the collection $\mathscr{C}$, 
where $|U_i|=\alpha_i$ denotes the number of packets stores on the node $U_i$, $\rho = \max\{\rho_j:j\in[\theta]\}$ is the maximum replication among all packets and the maximum number of packets on any node is given by $\alpha=\max\{\alpha_i:i\in[n]\}$. Clearly $\sum_{j=1}^{\theta}\rho_j = \sum_{i=1}^n\alpha_i$. 
\label{FR code}
\end{definition}

\begin{remark}
    For each $i\in[n]$ and each $j\in[\theta]$, the FR code $\mathscr{C}(n, \theta, \alpha, \rho)$ can also be represented as a collection of $\theta$ distinct $\rho_j$-element subsets of node set, where each node is the member of exactly $\alpha_i$ distinct subsets. 
\end{remark}

\begin{remark}
    An FR code $\mathscr{C}(n, \theta, \alpha, \rho)$ is called FR code with symmetric parameters (or uniform parameters) if $\alpha_i=\alpha$ ($i=1,2,\ldots,n$) and $\rho_i=\rho$ ($i=1,2,\ldots,\theta$), where a data collector can get complete file information by connecting any $k(<n)$ nodes and an arbitrary node can be repaired by any $d<n$ nodes of the FR code. All the FR codes, studied in \cite{rr10}, are the FR code swith symmetric parameters. An FR code with some non-uniform parameters is called FR code with asymmetric parameters.
\end{remark}

\begin{remark}
    In this paper, the notation $P_j\in U_i$ represents that packet $P_j$ is stored in node $U_i$. Again, the notation $P_j\notin U_i$ represents that packet $P_j$ is not stored in node $U_i$.
\end{remark}

An FR code with $n$ nodes and $\theta$ packets can also be represented by a matrix that is not necessarily a square matrix called \textit{node packet distribution incidence matrix} (NPDI Matrix). The entries of the matrix are associated with distribution of packets on nodes. 
Formally, the NPDI Matrix is defined as follows.
\begin{definition}(Node Packet Distribution Incidence Matrix):
    For an FR code  $\mathscr{C}:(n, \theta, \alpha, \rho)$, a matrix $M_{n\times\theta}=[a_{i,j}]$ over $\{0,1\}$ is called node packet distribution incidence matrix (NPDI Matrix) if 
    \[
    a_{i,j} = \left\{
    \begin{array}{ll}
         1\  & P_j\in U_i; \\
         0\  & P_j\notin U_i.
    \end{array}
    \right.
    \]
    \label{NPDI Matrix definition}
\end{definition}
For the FR code $\mathscr{C}:(6,10,4,3)$ as considered in the Table \ref{example}, the NPDI Matrix $M_{6\times 10}$ is as follows.
\begin{equation}
    M_{6\times 10} =
    \begin{bmatrix}
    1 & 1 & 1 & 1 & 0 & 0 & 0 & 0 & 0 & 0 \\
    1 & 0 & 0 & 0 & 1 & 1 & 1 & 0 & 0 & 0 \\
    0 & 1 & 0 & 0 & 1 & 0 & 0 & 1 & 1 & 0 \\
    0 & 0 & 1 & 0 & 0 & 1 & 0 & 1 & 0 & 0 \\
    0 & 0 & 0 & 1 & 0 & 0 & 1 & 0 & 0 & 1 \\
    0 & 0 & 0 & 0 & 0 & 0 & 0 & 1 & 1 & 1 \\
\end{bmatrix}_{6 \times 10}. 
\label{NPDI Matrix example}
\end{equation}


For an FR code $\mathscr{C}:(n,\theta,\alpha,\rho)$, let $\mathcal{P}_{\mathscr{C}}$ be the minimum number of distinct packets contained in the union of nodes of an arbitrary reconstruction set $\mathcal{A}$. Mathematically, 

    \[
    \mathcal{P}_{\mathscr{C}}:= \min_{\mathcal{A}}\left\lbrace \left| \bigcup_{U\in\mathcal{A}}U\right|\right\rbrace.
    \]
    By properties of a reconstruction set for an FR code, one can observe $B\leq\mathcal{P}_{\mathscr{C}}$.

For an FR code  $\mathscr{C}(n, \theta, \alpha, \rho)$, the amount $\mathcal{P}_{\mathscr{C}}$ is the maximum file size that can be delivered to a data collector by connecting nodes from any reconstructing set. For each reconstruction set $\mathcal{A}$, there is a $k_{max}$-element subset $\{U_i:i\in I\subset[n], |I|=k_{max}\}$ of node set such that $\mathcal{A}\subseteq\{U_i:i\in I\subset[n], |I|=k_{max}\}$, where $k_{max}$ is the maximum reconstruction degree. For given positive integer $k_{max}<n$, the code dimension $\mathcal{D}_{\mathscr{C}}(k)$ of an FR code $\mathscr{C}:(n,\theta,\alpha,\rho)$ is the maximum number of packets, which has guarantee to deliver to the data collector connected with any $k_{max}$ nodes. Formally, the code dimension $\mathcal{D}_{\mathscr{C}}(k)$ is defined as follows. 

\begin{definition}
    (Code Dimension): For a given integer $k<n$, the code dimension $\mathcal{D}_{\mathscr{C}}(k)$ of an FR code $\mathscr{C}:(n,\theta,\alpha,\rho)$ is the maximum number of distinct packets that guarantee to deliver to any user connected with any $k$ nodes of the FR code. Mathematically, 
    \[
    \mathcal{D}_{\mathscr{C}}(k):= \min_{\begin{array}{c}
                                        I\subset[n] \\
                                        |I|=k
                                      \end{array}
    }\left\lbrace \left| \bigcup_{i\in I}U_i\right|\right\rbrace,
    \]
    where $U_i$ ($i\in[n]$) is node on which packets are stored.
    Note that if the file with size $B$ can be retrieved by connecting any $k$ nodes then $B\leq\mathcal{D}_\mathscr{C}(k)$.
\end{definition}
For the FR code $\mathscr{C}:(6,10,4,3)$ as given in Table \ref{example}, one can find the code dimension $\mathcal{D}_{\mathscr{C}}(3)$ is $7$. 

In \cite{rr10}, an FR code $\mathscr{C}(n, \theta, \alpha, \rho)$ with a symmetric parameter is called \textit{universally good} if the code dimension is guaranteed to be more than or equal to the DSS capacity $kd-{k \choose 2}$ for any $k<n$ $i.e.$
\begin{equation}    
\mathcal{D}_{\mathscr{C}}(k) \geq kd -{k \choose 2}.
\label{uni-good on Homo FR code}
\end{equation}

\begin{remark}
    Consider an FR code $\mathscr{C}(n, \theta, \alpha, \rho)$ with 
    \begin{itemize}
        \item uniform node storage capacity $i.e.$ $\alpha_i = \alpha$ (for each $i\in[n]$) and
        \item $|U_i\cap U_j|\leq1$ for each $1\leq i<j\leq n$.
    \end{itemize} 
    The FR code satisfies the Inequality (\ref{uni-good on Homo FR code}) \cite{rnvr9a}.
\end{remark}

An FR code $\mathscr{C}(n, \theta, \alpha, \rho)$ with an asymmetric parameter is called \textit{universally good} if the code dimension is guaranteed to be more than or equal to $\sum_{i=1}^k d_i-{k \choose 2}$ for any $k<n$ $i.e.$
\begin{equation}    
\mathcal{D}_{\mathscr{C}}(k) \geq \sum_{i=1}^k d_i-{k \choose 2},
\label{uni-good on Hetero FR code}
\end{equation}
where $d_1\leq d_2\leq\ldots\leq d_n$ and $d_i$ is repair degree of node $U_i$.

\begin{remark}
    An FR code $\mathscr{C}(n, \theta, \alpha, \rho)$ with $|U_i\cap U_j|\leq1$ (for each $1\leq i<j\leq n$) satisfies the Inequality (\ref{uni-good on Hetero FR code}) \cite{6763122}.
    \label{remark 4}
\end{remark}

For a given $n\in\N$, the code rate of an ($n,k$) DSS (DSS with $n$ distinct nodes and the maximum reconstruction degree $k<n$) is 
    \[
    \mathcal{R}_{DSS}(k):=\frac{k}{n}.
    \]

Similarly, for an FR code defined on an $(n,k)$ DSS, the \textit{FR code rate} is the average fraction of information stored in a message packet. Formally the FR code rate can be defined as follows.

\begin{definition}(FR Code Rate):
    For an FR code $\mathscr{C}:(n,\theta,\alpha,\rho)$, the FR code rate is given by
    \begin{equation}
        \mathcal{R}_{\mathscr{C}}(k) := \frac{\mathcal{D}_{\mathscr{C}}(k)}{\sum_{j=1}^\theta\rho_j}.
    \end{equation}
\end{definition}
For the FR code as given in Table \ref{example}, one can find the FR code rate $\mathcal{R}_{\mathscr{C}}(3)$ that is $7/21$. 

\section{Bounds on Fractional Repetition Codes}
For FR codes, bounds on code dimension, FR code rate and DSS code rate are calculated in this section. The following theorem gives a bound on the code dimension. 

\begin{theorem}
 Let $\mathscr{C}:(n,\theta,\alpha,\rho)$ be an FR code with $n$ nodes and the maximum reconstruction degree $k$. If $\rho_j$ ($j\in[\theta]$) is the replication factor for packet $P_j$ then the maximum file size $\mathcal{D}_{\mathscr{C}}(k)$ stored in the FR code, is bounded as
 \begin{equation}
     \mathcal{D}_{\mathscr{C}}(k) \leq\left\lfloor\sum_{j=1}^\theta\left(1-\frac{{n-\rho_j\choose k}}{{n\choose k}}\right)\right\rfloor.
 \end{equation}
 \label{theorem 9}
\end{theorem}
\begin{proof}
The proof is motivated by the proof of the Lemma $(14)$ in \cite{rr10}. Let $\mathcal{S}$ = $\{S_I=\cup_{i\in I}U_i:I\subset[n], |I|=k\}$ be a collection of packet sets such that the packet set is contained in $k$-subsets of node set collectively. If the average cardinality of the sets in $\mathcal{S}$ is denoted by $\overline{S}$, then 
\begin{equation}
    \sum_{S_I\in\mathcal{S}}|S_I|={n\choose k}\overline{S}.
\end{equation}
But, an arbitrary packet $P_j$ ($j\in[\theta]$) is the member of exactly ${n\choose k}-{n-\rho_j\choose k}$ distinct sets in $\mathcal{S}$. Hence,
\begin{equation}
    \sum_{S_I\in\mathcal{S}}|S_I|=\sum_{j=1}^\theta\left({n\choose k}-{n-\rho_j\choose k}\right).
\end{equation}
Since, any $k$ nodes have sufficient packets to reconstruct file so, $\mathcal{D}_{\mathscr{C}}(k) \leq |S_I|$ for each $S_I\in\mathcal{S}$. Hence,
\begin{equation}
    \mathcal{D}_{\mathscr{C}}(k) \leq \overline{S} = \left \lfloor\sum_{j=1}^\theta\left(1-\frac{{n-\rho_j\choose k}}{{n\choose k}}\right)\right \rfloor.
\end{equation}
\end{proof}

 \begin{remark}
For an FR code with symmetric parameters, the Theorem \ref{theorem 9} is proved in \cite{rr10}. For an FR code with symmetric node storage capacity and symmetric reconstruction degree, the Theorem \ref{theorem 9} is proved in \cite[Equation 6]{iet:/content/journals/10.1049/iet-com.2014.1225}. 
\end{remark}

If a data collector can retrieve the complete file information form the FR code in presence of an arbitrary node failure then the FR code is reliable FR code. The following lemma gives a bound on FR code rate of a reliable FR code.

\begin{lemma}
    The code rate of a reliable FR code is bounded by $1/2$.
    \label{lemma 1 on code rate}
\end{lemma}
\begin{proof}
    If an FR code $\mathscr{C}(n, \theta, \alpha, \rho)$ is reliable, then the replication factor of each packet is atleast $2$. Hence, $\sum_{j=1}^\theta\rho_j\geq 2\theta$. The inequality follows 
    \[
    \mathcal{R}_{\mathscr{C}}(k) =\frac{\mathcal{D}_{\mathscr{C}}(k)}{\sum_{j=1}^\theta\rho_j} \leq \frac{\mathcal{D}_{\mathscr{C}}(k)}{2\theta}.
    \]
    The code dimension $\mathcal{D}_{\mathscr{C}}(k)$ can not be more than the total number of encoded packets $\theta$. Therefore,
    
    \[
    \mathcal{R}_{\mathscr{C}}(k) =\frac{\mathcal{D}_{\mathscr{C}}(k)}{\sum_{j=1}^\theta\rho_j} \leq \frac{\mathcal{D}_{\mathscr{C}}(k)}{2\theta}\leq \frac{1}{2}.
    \]
    It follows the lemma.
\end{proof}

Consider an FR code $\mathscr{C}:(n,\theta,\alpha,\rho)$ with replication factor $\rho_j=2$ (for each $j\in[\theta]$) on a $(\theta,B)$ MDS code. The FR code rate approaches to $1/2$, if $\mathcal{D}_{\mathscr{C}}(k)/\theta$ approaches to $1$. Note that the ratio $\mathcal{D}_{\mathscr{C}}(k)/\theta$ will approach $1$, if the values of both the code dimension $\mathcal{D}_{\mathscr{C}}(k)$ and the total number of distinct encoded packets $\theta$ are very close and are very high. Recall that the tolerance factor of the FR code is $\rho-1=1$.

\begin{remark}
    For an FR code $\mathscr{C}:(n,\theta,\alpha,\rho)$, the code rate     
    \[
    \mathcal{R}_{\mathscr{C}}(k) =\frac{\mathcal{D}_{\mathscr{C}}(k)}{\sum_{j=1}^\theta\rho_j} \geq \frac{\alpha_{min}}{\sum_{j=1}^\theta\rho_j}\geq \frac{\alpha_{min}}{\rho\theta}.
    \]   
    \label{lower bound on FR code rate}
\end{remark}

\begin{lemma}
    For any FR code $\mathscr{C}:(n,\theta,\alpha,\rho)$, the growth rate of the FR code rate is
    \[
    \frac{\mathcal{R}_{\mathscr{C}}(k)-\mathcal{R}_{\mathscr{C}}(k-1)}{\mathcal{R}_{\mathscr{C}}(k-1)} \leq \frac{\sum_{j=1}^n\alpha_i}{2\alpha_{min}}-1,
    \]
    where $k=2,3,\ldots n-1$.
\end{lemma}
\begin{proof}
    The proof follows the Lemma \ref{lemma 1 on code rate} and the Remark \ref{lower bound on FR code rate}.
\end{proof}

\begin{remark}
    For any FR code $\mathscr{C}:(n,\theta,\alpha,\rho)$ with symmetric parameters, the growth rate of FR code rate is
    \[
    \frac{\mathcal{R}_{\mathscr{C}}(k)-\mathcal{R}_{\mathscr{C}}(k-1)}{\mathcal{R}_{\mathscr{C}}(k-1)} \leq \frac{n}{2}-1.
    \]
\end{remark}

    For an FR code $\mathscr{C}:(n,\theta,\alpha,\rho)$, $\sum_{j=1}^{\theta}\rho_j = \sum_{i=1}^n\alpha_i$. So, the FR code rate 
    \[
    \mathcal{R}_{\mathscr{C}}(k) = \frac{\mathcal{D}_{\mathscr{C}}(k)}{\sum_{i=1}^n\alpha_i}.
    \] 
    
    For each reconstruction degree $k<n$, each FR code $\mathscr{C}:(n,\theta,\alpha,\rho)$ is defined on an ($n,k$) DSS. 
    The following lemma gives a relation between the FR code rate and the respective DSS code rate.

\begin{lemma}
    For an FR code, the FR code rate is bounded by the DSS code rate.
    \label{lemma 3 on code rate}
\end{lemma}
\begin{proof}
    For an FR code, let a file with size $\mathcal{D}_{\mathscr{C}}(k)$ (equal to the code dimension), be stored in an FR code $\mathscr{C}:(n,\theta,\alpha,\rho)$. Since $\frac{1}{k}\mathcal{D}_{\mathscr{C}}(k)\leq\frac{1}{n}\sum_{i=1}^n\alpha_i$,
    \[
    \mathcal{R}_{\mathscr{C}}(k) = \frac{\mathcal{D}_{\mathscr{C}}(k)}{\sum_{i=1}^n\alpha_i} \leq 
    \frac{k}{n}\left(\frac{\frac{1}{k}\mathcal{D}_{\mathscr{C}}(k)}{\frac{1}{n}\sum_{i=1}^n\alpha_i}\right)<\frac{k}{n} = \mathcal{R}_{DSS}(k).
    \]
\end{proof}

\begin{theorem}
    For an ($n,k$) DSS, consider an FR code $\mathscr{C}:(n,\theta,\alpha,\rho)$ with $|U_i\cap U_j|\leq1$ (for each $1\leq i<j\leq n$). Then, the rate difference
    \[
    \mathcal{R}_{DSS}(k)-\mathcal{R}_{\mathscr{C}}(k)\leq\frac{1}{\sum_{i=1}^n\alpha_i}{k\choose 2}.
    \]
\end{theorem}
\begin{proof}
For any $i,j\in[n]$ and $1\leq i<j\leq n$, let an FR code $\mathscr{C}:(n,\theta,\alpha,\rho)$ with $|U_i\cap U_j|\leq1$ be defined on ($n,k$) DSS. If $\alpha_{ave}$ is the average node storage capacity on nodes, then the code rate difference 
\begin{equation*}
\begin{split}    
    \mathcal{R}_{DSS}(k)-\mathcal{R}_{\mathscr{C}}(k) 
    & = \frac{k\frac{1}{n}\sum_{i=1}^n\alpha_i-\mathcal{D}_{\mathscr{C}}(k)}{\sum_{i=1}^n\alpha_i} \\ 
    & =\frac{k\alpha_{ave}-\mathcal{D}_{\mathscr{C}}(k)}{\sum_{i=1}^n\alpha_i}.
\end{split}
\end{equation*}

Without loss of generality, let $\alpha_i\geq\alpha_j$ for $1\leq i<j\leq n$. Hence, 
\[
    k\alpha_{ave}\leq \sum_{i=1}^k\alpha_i.
\] Therefore,
\[
    \mathcal{R}_{DSS}(k)-\mathcal{R}_{\mathscr{C}}(k)\leq \frac{1}{\sum_{i=1}^n\alpha_i}\left(\sum_{i=1}^k\alpha_i-\mathcal{D}_{\mathscr{C}}(k)\right).
\]
Note that $\sum_{i=1}^{n}\alpha_i-\mathcal{D}_{\mathscr{C}}(k)$ are the total number of common packets such that every common packet is shared by a pair of the nodes from the node set $\{U_1$, $U_2$, $\ldots$, $U_k\}$. Hence,
\[
    \mathcal{R}_{DSS}(k)-\mathcal{R}_{\mathscr{C}}(k) \leq \frac{1}{\sum_{i=1}^n\alpha_i}{k\choose 2}.
\]
\end{proof}

Consider an FR code $\mathscr{C}:(5,10,4,2)$ with the node packet distribution as given in the Table \ref{5,10,4,2}. For $k=1,2,3,4$, the FR code rate $\mathcal{R}_{\mathscr{C}}(k)$, DSS code rate $k/n$ and difference between both the code rates ate calculated in Table \ref{rates}.

\begin{table}[ht]
    \caption{The Node Packet Distribution for the FR code $\mathscr{C}:(5,10,4,2)$.}
    \centering 
    \begin{tabular}{|c||c|}
        \hline
        $U_1$& $P_1$, $P_2$, $P_3$, $P_4$         \\
        \hline \hline
        $U_2$& $P_1$, $P_5$, $P_6$, $P_7$         \\
        \hline \hline
        $U_3$& $P_2$, $P_5$, $P_8$, $P_9$         \\
        \hline \hline 
        $U_4$& $P_3$, $P_6$, $P_8$, $P_{10}$          \\
        \hline \hline 
        $U_5$& $P_4$, $P_7$, $P_9$, $P_{10}$          \\
        \hline  
    \end{tabular}
    \label{5,10,4,2}
\end{table}

\begin{table}[ht]
    \caption{For $k=2,3,4$, the ($5,k$) DSS code rate and the FR code rate of the FR code $\mathscr{C}:(5,10,4,2)$.}
    \centering 
    \begin{tabular}{|c|c|c|c|}
        \hline
        Max. Reconstruction & DSS Code Rate   &  Code Rate  & The Rate   \\
         Degree             &                 &  of FR Code & Difference \\
        $k$                 & $\mathcal{R}_{DSS}(k)$   &  $\mathcal{R}_{\mathscr{C}}(k)$  & $k(k-1)/(2n\alpha)$ \\
        \hline \hline
        $1$& $0.200$    &   $0.200$ & $0.000$\\
        \hline 
        $2$& $0.400$    &   $0.350$ & $0.050$\\
        \hline 
        $3$& $0.600$     &  $0.450$ & $0.150$ \\
        \hline 
        $4$& $0.800$   &    $0.500$  & $0.300$ \\
        \hline  
    \end{tabular}
    \label{rates}
\end{table}

For another example, consider an FR code $\mathscr{C}:(6,6,2,2)$ with the node packet distribution as given in the Table \ref{6,6,2,2}. For $k=2,3,4,5$, the code rate $\mathcal{R}_{\mathscr{C}}(k)$, the DSS code rate $\mathcal{R}_{DSS}(k)$ and the difference $\mathcal{R}_{DSS}(k)-\mathcal{R}_{\mathscr{C}}(k)$ are calculated in Table \ref{rates 6,6,2,2}.

\begin{table}[ht]
    \caption{The Node Packet Distribution for the FR code $\mathscr{C}:(6,6,2,2)$.}
    \centering 
    \begin{tabular}{|c||c|}
        \hline
        $U_1$& $P_1$, $P_2$        \\
        \hline \hline
        $U_2$& $P_2$, $P_3$         \\
        \hline \hline
        $U_3$& $P_3$, $P_4$         \\
        \hline \hline 
        $U_4$& $P_4$, $P_5$          \\
        \hline \hline 
        $U_5$& $P_5$, $P_6$          \\
        \hline \hline 
        $U_6$& $P_1$, $P_6$          \\
        \hline  
    \end{tabular}
    \label{6,6,2,2}
\end{table}

\begin{table}[ht]
    \caption{For $k=1,2,3,4,5$, the ($6,k$) DSS code rate and the FR code rate of an FR code $\mathscr{C}:(6,6,2,2)$.}
    \centering 
    \begin{tabular}{|c|c|c|c|}
        \hline
        Max. Reconstruction & DSS Code Rate   & FR Code Rate & The Rate   \\
         Degree             &                 &              & Difference \\
        $k$& $\mathcal{R}_{DSS}(k)$   &  $\mathcal{R}_{\mathscr{C}}(k)$  & $(k-1)/(2n)$ \\
        \hline \hline
        $1$& $1/6$    &   $2/12$ & $0$     \\
        \hline 
        $2$& $2/6$    &   $3/12$ & $1/12$  \\
        \hline 
        $3$& $3/6$     &  $4/12$ & $2/12$  \\
        \hline 
        $4$& $4/6$   &    $5/12$  & $3/12$ \\
        \hline 
        $5$& $5/6$   &    $6/12$  & $4/12$ \\
        \hline  
    \end{tabular}
    \label{rates 6,6,2,2}
\end{table}

\begin{lemma}
    For an FR code, if the FR code rate and the DSS code rate are same at some $k>1$ then the FR code is not reliable.
\end{lemma}
\begin{proof}
    Consider an FR code $\mathscr{C}:(n,\theta,\alpha,\rho)$ which is defined on an ($n,k$) DSS. For any given $k<n$, if possible then assume $\mathcal{R}_{\mathscr{C}}(k) = k/n$. It implies  
    \[
    \frac{\mathcal{D}_{\mathscr{C}}(k)}{k} = \frac{1}{n}\sum_{i=1}^n\alpha_i.
    \]
    So, on an average the total number of information packets are equal to the node storage capacity. Hence, there is not any space for replicated packets and it is not the case of a reliable FR code. Hence the lemma has proved.
\end{proof}

\subsection{Construction of Universally Good Fractional Repetition Codes}
  
 In this subsection, we give a construction of universally good FR code by concatenating two FR codes with different packet symbols. The resultant FR code can be called concatenated FR code.
 
 Let $\mathscr{C}^{(1)}:(n^{(1)}, \theta^{(1)}, \alpha^{(1)}, \rho^{(1)})$ and $\mathscr{C}^{(2)}:(n^{(2)}, \theta^{(2)}, \alpha^{(2)}, \rho^{(2)})$ be two FR codes with NPDI Matrices $M^{(1)}$ and $M^{(2)}$ respectively. Let both the FR codes have different packet symbols. An FR code with NPDI Matrix  
  \[
    M = \left[
    \begin{array}{cc}
    M^{(1)} & \textbf{0} \\ 
    \textbf{0} & M^{(2)}
    \end{array}\right]
    \]
  is called concatenated FR code $\mathscr{C}:(n, \theta, \alpha, \rho)$ with $n=n^{(1)}+n^{(2)}$, $\theta=\theta^{(1)}+\theta^{(2)}$, $\alpha=\max\{\alpha^{(1)},\alpha^{(2)}\}$ and $\rho=\max\{\rho^{(1)},\rho^{(2)}\}$. In the NPDI Matrix $M$, the matrix \textbf{0} denotes a zero matrix and both the zero matrices may have different dimensions.
  
The following two theorems give the relation between the code rates of concatenated FR codes. 

\begin{theorem}
    For any $n\in\N$, consider two FR codes $\mathscr{C}^{(1)}:(n^{(1)}, \theta^{(1)}, \alpha^{(1)}, \rho^{(1)})$ and $\mathscr{C}^{(2)}:(n^{(2)}, \theta^{(2)}, \alpha^{(2)}, \rho^{(2)})$ with NPDI Matrices $M^{(1)}$ and $M^{(2)}$ respectively. Let an FR code $\mathscr{C}:(n, \theta, \alpha, \rho)$ be constructed with the NPDI Matrix 
    \[
    M = \left[
    \begin{array}{cc}
    M^{(1)} & \textbf{0} \\ 
    \textbf{0} & M^{(2)}
    \end{array}\right]
    \]
    by concatenating the two FR codes $\mathscr{C}^{(1)}$ and $\mathscr{C}^{(2)}$. For any $k<n=n^{(1)}+n^{(2)}$, 
            \begin{equation*}
            \begin{split}
            \frac{1}{\mathcal{R}_{DSS}(n,k)} = & \frac{1}{\left\lfloor\frac{k}{n^{(1)}}\right\rfloor+\mathcal{R}_{DSS}(n^{(1)},k)} \\ 
            &  + \frac{1}{\left\lfloor\frac{k}{n^{(2)}}\right\rfloor+\mathcal{R}_{DSS}(n^{(2)},k)},
            \end{split}
            \end{equation*} and the FR code rate
            \[
            \mathcal{R}_{\mathscr{C}}(k) \leq 
            \left\{
            \begin{array}{ll}
                 r^{(1)}\mathcal{R}_{\mathscr{C}^{(1)}}(k) + r^{(2)}\mathcal{R}_{\mathscr{C}^{(2)}}(k), & k<n_{min}; \\
                 \varrho, & \mbox{ otherwise};
            \end{array}
            \right.
            \]
            where $n_{min} = \min\{n^{(1)},n^{(2)}\}$, $\varrho = \min\{r^{(i)} + r^{(j)}\mathcal{R}_{\mathscr{C}^{(j)}}(k): i,j\in[2], i+j = 3\}$, $\mathcal{R}_{\mathscr{C}}(0) = 0$, $\mathcal{R}_{\mathscr{C}^{(j)}}(1) = \alpha^{(j)}_{min}/\sum_{i=1}^{n^{(j)}}\alpha_i^{(j)}$ and
            \[
            r^{(j)} = \frac{\sum_{i=1}^{n^{(j)}}\alpha_i^{(j)}}{\sum_{i=1}^{n^{(1)}}\alpha_i^{(1)}+\sum_{i=1}^{n^{(2)}}\alpha_i^{(2)}}
            \]
            for $j=1,2$.
    \label{FR code concatenation diagonal}
\end{theorem}

\begin{proof}
    The proof is straightforward.
\end{proof}

\begin{theorem}
    Consider an FR code $\mathscr{C}^{(1)}:(n^{(1)}, \theta^{(1)}, \alpha^{(1)}, \rho^{(1)})$ with NPDI Matrix $M^{(1)}_{n\times\theta}$. For a positive integer $m$ and $r<m$, let $\mathscr{C}^{(m)}:(n^{(m)}=mn, \theta^{(m)}=m\theta, \alpha^{(1)}, \rho^{(m)}=m\rho)$ be an FR code with NPDI Matrix $M^{(m)}$, where the matrix 
    \[
    M^{(r+1)} =
    \begin{bmatrix}
    M^{(1)} & \textbf{0} \\
    \textbf{0} & M^{(r)} \\
    \end{bmatrix}.
    \]
     For a positive integer $K<n^{(m)}$ and $k\equiv K\pmod{n^{(1)}}$, 
    \begin{enumerate}
        \item the DSS code rate 
            \[
            \mathcal{R}_{DSS}(K) = \frac{1}{m}\left(\left\lfloor\frac{K}{n^{(m)}}\right\rfloor+\frac{k}{n}\right),
            \]
        \item the FR code rate
            \[
            \mathcal{R}_{\mathscr{C}^{(m)}}(K) = \frac{1}{m}\left(\left\lfloor\frac{K}{n^{(m)}}\right\rfloor\frac{1}{\rho_{ave}}+\mathcal{R}_{\mathscr{C}}(k)\right),
            \]
        \item and the rate difference
            \begin{equation*}
            \begin{split}
            \mathcal{R}_{DSS}(K) & - \mathcal{R}_{\mathscr{C}^{(n)}}(K) \\ 
            & = \frac{1}{m}\left\{\mathcal{R}_{DSS}(k)\left(1-\frac{1}{\rho_{ave}}\right)-\mathcal{R}_{\mathscr{C}}(k)\right\},
            \end{split}
            \end{equation*}
    \end{enumerate}
    where $\mathcal{R}_{DSS}(k)$ is ($n^{(1)},k$) DSS code rate and 
    \[
    \rho_{ave}  = \frac{\sum_{i=1}^n\alpha_i}{\theta} = \frac{\sum_{j=1}^\theta\rho_j}{\theta}.
    \]
    \label{rate difference on concatenated FR code}
\end{theorem}

\begin{proof}
    For $N=nm$ and $K=n\lfloor K/n\rfloor+k$, one can easily find the DSS code rate $\mathcal{R}_{DSS}(K)$, the FR code rate $\mathcal{R}_{\mathscr{C}(K)}$ and the rate difference $\mathcal{R}_{DSS}(K) - \mathcal{R}_{\mathscr{C}'}(K)$.
\end{proof}

In an FR code $\mathscr{C}^{(m)}:(n^{(m)}=mn, \theta^{(m)}=m\theta, \alpha^{(1)}, \rho^{(m)}=m\rho)$, if the average of the replication factor to the packets is 
\[
\rho_{ave}^{(m)} = \frac{\sum_{i=1}^{n^{(m)}}\alpha^{(m)}_i}{\theta^{(m)}},
\] 
then $\rho_{ave}^{(m)}=\rho_{ave}^{(1)}$, for any $m\in\N$.

\begin{remark}
    For $r=1,2$, if two FR codes $\mathscr{C}^{(r)}:(n^{(r)}, \theta^{(r)}, \alpha^{(r)}, \rho^{(r)})$ are universally good then the concatenated FR code $\mathscr{C}:(n, \theta, \alpha, \rho)$ (see the Theorem \ref{FR code concatenation diagonal}) is also universally good, where  $n=n^{(1)}+n^{(2)}$, $\theta = \theta^{(1)} + \theta^{(2)}$, $\alpha = \min\{\alpha^{(r)}: r=1,2\}$ and $\rho = \max\{\rho^{(r)}:r=1,2\}$.
\end{remark}

\begin{remark}
    For any $m\in\N$, if the FR code $\mathscr{C}^{(1)}:(n^{(1)}, \theta^{(1)}, \alpha^{(1)}, \rho^{(1)})$ is universally good then the constructed FR code $\mathscr{C}^{(m)}:(n^{(m)}, \theta^{(m)}, \alpha^{(m)}, \rho^{(m)})$ (see the Theorem \ref{rate difference on concatenated FR code}) is also universally good, where $n^{(m)}=mn, \theta^{(m)}=m\theta, \alpha^{(m)}=\alpha^{(1)}$ and $\rho^{(m)}=m\rho$.
\end{remark}

\begin{remark}
    For a positive integer $\alpha>1$, consider a matrix 
    \[
    W_{\alpha+1} = 
    \begin{bmatrix}
        \textbf{1}_\alpha & \textbf{0}_{\alpha\choose 2}  \\
         I_\alpha & W_\alpha
\end{bmatrix},
    \]
     where $\textbf{0}_m$ is an zero-array with length $m$, $\textbf{1}_m$ is an $m$-length array with each entry one, $I_m$ is an identity matrix of dimension $m$ and $W_1$ 
     = $\begin{bmatrix}
        1  \\
        1
\end{bmatrix}$ for $m>1$ and $m\in\N$.
     An FR code $\mathscr{C}:(n=\alpha +1, \theta={\alpha +1\choose 2}, \alpha, \rho = 2)$ with the NPDI Matrix $W_{\alpha+1}$ is universally good \cite{rr10}. For the FR code, the difference between the DSS code rate $\mathcal{R}_{DSS}(k)$ and the FR code rate $\mathcal{R}_\mathscr{C}(k)$ is $k(k-1)/(2n\alpha)$. Note that the difference increases with $\mathcal{O}(k^2)$. 
     \end{remark}

\section{Conclusion}
In this work, FR codes with non-uniform parameters are considered. For such FR codes, bounds on the DSS code rate, the FR code rate and the growth of the FR code rate are investigated, where the DSS code rate is the fraction of the information per node and the FR code rate is the fraction of the information per encoded packet. These bounds are also obtained for the concatenated FR codes. 

\bibliographystyle{IEEEtran}
\bibliography{cloud}

\begin{thebibliography}{10}
\providecommand{\url}[1]{#1}
\csname url@samestyle\endcsname
\providecommand{\newblock}{\relax}
\providecommand{\bibinfo}[2]{#2}
\providecommand{\BIBentrySTDinterwordspacing}{\spaceskip=0pt\relax}
\providecommand{\BIBentryALTinterwordstretchfactor}{4}
\providecommand{\BIBentryALTinterwordspacing}{\spaceskip=\fontdimen2\font plus
\BIBentryALTinterwordstretchfactor\fontdimen3\font minus
  \fontdimen4\font\relax}
\providecommand{\BIBforeignlanguage}[2]{{%
\expandafter\ifx\csname l@#1\endcsname\relax
\typeout{** WARNING: IEEEtran.bst: No hyphenation pattern has been}%
\typeout{** loaded for the language `#1'. Using the pattern for}%
\typeout{** the default language instead.}%
\else
\language=\csname l@#1\endcsname
\fi
#2}}
\providecommand{\BIBdecl}{\relax}
\BIBdecl

\bibitem{5550492}
A.~Dimakis, P.~Godfrey, Y.~Wu, M.~Wainwright, and K.~Ramchandran, ``Network
  coding for distributed storage systems,'' \emph{Information Theory, IEEE
  Transactions on}, vol.~56, no.~9, pp. 4539--4551, Sept 2010.

\bibitem{DARKYC11}
A.~Dimakis, K.~Ramchandran, Y.~Wu, and C.~Suh, ``A survey on network codes for
  distributed storage,'' \emph{Proceedings of the IEEE}, vol.~99, no.~3, pp.
  476 --489, march 2011.

\bibitem{6620424}
J.~Pernas, C.~Yuen, B.~Gaston, and J.~Pujol, ``Non-homogeneous two-rack model
  for distributed storage systems,'' in \emph{Information Theory Proceedings
  (ISIT), 2013 IEEE International Symposium on}, July 2013, pp. 1237--1241.

\bibitem{5961826}
K.~Rashmi, N.~Shah, and P.~Kumar, ``Optimal exact-regenerating codes for
  distributed storage at the {MSR} and {MBR} points via a product-matrix
  construction,'' \emph{Information Theory, IEEE Transactions on}, vol.~57,
  no.~8, pp. 5227 --5239, aug. 2011.

\bibitem{RSKR09}
K.~Rashmi, N.~Shah, P.~Kumar, and K.~Ramchandran, ``Explicit construction of
  optimal exact regenerating codes for distributed storage,'' in
  \emph{Communication, Control, and Computing, 2009. Allerton 2009. 47th Annual
  Allerton Conference on}, 30 2009-oct. 2 2009, pp. 1243 --1249.

\bibitem{5513367}
------, ``Explicit and optimal exact-regenerating codes for the
  minimum-bandwidth point in distributed storage,'' in \emph{Information Theory
  Proceedings (ISIT), 2010 IEEE International Symposium on}, June 2010, pp.
  1938 --1942.

\bibitem{5513263}
C.~Suh and K.~Ramchandran, ``Exact-repair {MDS} codes for distributed storage
  using interference alignment,'' in \emph{Information Theory Proceedings
  (ISIT), 2010 IEEE International Symposium on}, June 2010, pp. 161 --165.

\bibitem{rr10}
S.~El~Rouayheb and K.~Ramchandran, ``Fractional repetition codes for repair in
  distributed storage systems,'' in \emph{Communication, Control, and Computing
  (Allerton), 2010 48th Annual Allerton Conference on}, Oct. 2010, pp. 1510
  --1517.

\bibitem{DBLP:journals/corr/abs-1201-3547}
T.~Ernvall, ``The existence of fractional repetition codes,'' \emph{CoRR}, vol.
  abs/1201.3547, 2012.

\bibitem{6120326}
J.~Koo and J.~Gill, ``Scalable constructions of fractional repetition codes in
  distributed storage systems,'' in \emph{Communication, Control, and Computing
  (Allerton), 2011 49th Annual Allerton Conference on}, Sept 2011, pp.
  1366--1373.

\bibitem{6483351}
O.~Olmez and A.~Ramamoorthy, ``Repairable replication-based storage systems
  using resolvable designs,'' in \emph{Communication, Control, and Computing
  (Allerton), 2012 50th Annual Allerton Conference on}, Oct 2012, pp.
  1174--1181.

\bibitem{7422071}
------, ``Fractional repetition codes with flexible repair from combinatorial
  designs,'' \emph{IEEE Transactions on Information Theory}, vol.~62, no.~4,
  pp. 1565--1591, April 2016.

\bibitem{6810361}
------, ``Constructions of fractional repetition codes from combinatorial
  designs,'' in \emph{Signals, Systems and Computers, 2013 Asilomar Conference
  on}, Nov 2013, pp. 647--651.

\bibitem{6033980}
S.~Pawar, N.~Noorshams, S.~El~Rouayheb, and K.~Ramchandran, ``{DRESS} codes for
  the storage cloud: Simple randomized constructions,'' in \emph{Information
  Theory Proceedings (ISIT), 2011 IEEE International Symposium on}, July 2011,
  pp. 2338--2342.

\bibitem{DBLP:journals/corr/abs-1303-6801}
S.~Anil, M.~K. Gupta, and T.~A. Gulliver, ``Enumerating some fractional
  repetition codes,'' \emph{CoRR}, vol. abs/1303.6801, 2013.

\bibitem{6620277}
G.~M. Kamath, N.~Silberstein, N.~Prakash, A.~S. Rawat, V.~Lalitha, O.~O.
  Koyluoglu, P.~V. Kumar, and S.~Vishwanath, ``Explicit mbr all-symbol locality
  codes,'' in \emph{2013 IEEE International Symposium on Information Theory},
  July 2013, pp. 504--508.

\bibitem{6846301}
G.~M. Kamath, N.~Prakash, V.~Lalitha, and P.~V. Kumar, ``Codes with local
  regeneration and erasure correction,'' \emph{IEEE Transactions on Information
  Theory}, vol.~60, no.~8, pp. 4637--4660, Aug 2014.

\bibitem{DBLP:journals/corr/abs-1302-3681}
\BIBentryALTinterwordspacing
M.~K. Gupta, A.~Agrawal, and D.~Yadav, ``On weak dress codes for cloud
  storage,'' \emph{CoRR}, vol. abs/1302.3681, 2013. [Online]. Available:
  \url{http://arxiv.org/abs/1302.3681}
\BIBentrySTDinterwordspacing

\bibitem{7118709}
N.~Silberstein and T.~Etzion, ``Optimal fractional repetition codes based on
  graphs and designs,'' \emph{Information Theory, IEEE Transactions on},
  vol.~PP, no.~99, pp. 1--1, 2015.

\bibitem{6804948}
Q.~Yu, C.~W. Sung, and T.~Chan, ``Irregular fractional repetition code
  optimization for heterogeneous cloud storage,'' \emph{Selected Areas in
  Communications, IEEE Journal on}, vol.~32, no.~5, pp. 1048--1060, May 2014.

\bibitem{iet:/content/journals/10.1049/iet-com.2014.1225}
\BIBentryALTinterwordspacing
B.~Zhu, H.~Li, K.~W. Shum, and S.-Y.~R. Li,
  ``\BIBforeignlanguage{English}{{HFR} code: a flexible replication scheme for
  cloud storage systems},'' \emph{\BIBforeignlanguage{English}{IET
  Communications}}, October 2015. [Online]. Available:
  \url{http://digital-library.theiet.org/content/journals/10.1049/iet-com.2014.1225}
\BIBentrySTDinterwordspacing

\bibitem{6763122}
B.~Zhu, K.~Shum, H.~Li, and H.~Hou, ``General fractional repetition codes for
  distributed storage systems,'' \emph{Communications Letters, IEEE}, vol.~18,
  no.~4, pp. 660--663, April 2014.

\bibitem{6811237}
B.~Zhu, H.~Li, H.~Hou, and K.~W. Shum, ``Replication-based distributed storage
  systems with variable repetition degrees,'' in \emph{2014 Twentieth National
  Conference on Communications (NCC)}, Feb 2014, pp. 1--5.

\bibitem{7458383}
K.~G. Benerjee and M.~K. Gupta, ``On dress codes with flowers,'' in \emph{2015
  Seventh International Workshop on Signal Design and its Applications in
  Communications (IWSDA)}, Sept 2015, pp. 108--112.

\bibitem{DBLP:journals/corr/PrajapatiG16}
\BIBentryALTinterwordspacing
S.~A. Prajapati and M.~K. Gupta, ``On some universally good fractional
  repetition codes,'' \emph{CoRR}, vol. abs/1609.03106, 2016. [Online].
  Available: \url{http://arxiv.org/abs/1609.03106}
\BIBentrySTDinterwordspacing

\bibitem{7066224}
B.~Zhu, K.~W. Shum, and H.~Li, ``Heterogeneity-aware codes with uncoded repair
  for distributed storage systems,'' \emph{IEEE Communications Letters},
  vol.~19, no.~6, pp. 901--904, June 2015.

\bibitem{7312417}
B.~Zhu and H.~Li, ``Adaptive fractional repetition codes for dynamic storage
  systems,'' \emph{IEEE Communications Letters}, vol.~19, no.~12, pp.
  2078--2081, Dec 2015.

\bibitem{e18120441}
\BIBentryALTinterwordspacing
H.~Park and Y.-S. Kim, ``Construction of fractional repetition codes with
  variable parameters for distributed storage systems,'' \emph{Entropy},
  vol.~18, no.~12, 2016. [Online]. Available:
  \url{http://www.mdpi.com/1099-4300/18/12/441}
\BIBentrySTDinterwordspacing

\bibitem{7366761}
B.~Zhu, ``Rethinking fractional repetition codes: New construction and code
  distance,'' \emph{IEEE Communications Letters}, vol.~20, no.~2, pp. 220--223,
  Feb 2016.

\bibitem{7005805}
N.~Silberstein and T.~Etzion, ``Optimal fractional repetition codes for
  distributed storage systems,'' in \emph{2014 IEEE 28th Convention of
  Electrical Electronics Engineers in Israel (IEEEI)}, Dec 2014, pp. 1--4.

\bibitem{7458387}
M.~Y. Nam, J.~H. Kim, and H.~Y. Song, ``Locally repairable fractional
  repetition codes,'' in \emph{2015 Seventh International Workshop on Signal
  Design and its Applications in Communications (IWSDA)}, Sept 2015, pp.
  128--132.

\bibitem{7558231}
B.~Zhu and H.~Li, ``Exploring node repair locality in fractional repetition
  codes,'' \emph{IEEE Communications Letters}, vol.~20, no.~12, pp. 2350--2353,
  Dec 2016.

\bibitem{DBLP:journals/corr/abs-1305-4580}
\BIBentryALTinterwordspacing
K.~G. Benerjee, M.~K. Gupta, and N.~Agrawal, ``Reconstruction and repair degree
  of fractional repetition codes,'' \emph{CoRR}, vol. abs/1305.4580, 2013.
  [Online]. Available: \url{http://arxiv.org/abs/1305.4580}
\BIBentrySTDinterwordspacing

\bibitem{rnvr9a}
K.~Rashmi, N.~Shah, P.~Kumar, and K.~Ramchandran, ``Explicit construction of
  optimal exact regenerating codes for distributed storage,'' in
  \emph{Communication, Control, and Computing, 2009. Allerton 2009. 47th Annual
  Allerton Conference on}, Oct. 2009, pp. 1243 --1249.

\end{thebibliography}

\end{document}